\documentclass[a4paper,UKenglish,cleveref, autoref, thm-restate]{lipics-v2021}

\pdfoutput=1 
\nolinenumbers 
\hideLIPIcs  

\usepackage{stmaryrd}

\usepackage{anyfontsize}
\SetSymbolFont{stmry}{bold}{U}{stmry}{m}{n}

\usepackage{tikz}
\usepackage{dsfont} 
\usepackage{mathtools} 
\newcommand{\quickset}[1]{\left\lbrace #1 \right\rbrace}

\newtheorem{fact}{Fact}[section]

\newcommand{\segment}[2]{\llbracket #1, #2 \rrbracket}
\newcommand{\cX}{\mathcal{X}}

\newcounter{algorithm}
\newtheorem{algo}[algorithm]{Algorithm}{\bfseries}{\itshape}

\title{Self-Stabilization and Byzantine Tolerance for Maximal Independent Set} 

\author{Johanne Cohen}{LISN-CNRS, Universit\'e Paris-Saclay, France}{Johanne.Cohen@lri.fr}{https://orcid.org/0000-0002-9548-5260}{}

\author{Laurence Pilard}{LI-PaRAD,  UVSQ, Universit\'e Paris-Saclay, France,}{Laurence.Pilard@uvsq.fr}{https://orcid.org/0000-0002-1104-8216}{}

\author{François Pirot}{LISN-CNRS, Universit\'e Paris-Saclay, France}{Francois.Pirot@lri.fr}{https://orcid.org/0000-0002-1392-9623}{}
\authorrunning{J. Cohen, L. Pilard,  F. Pirot and  J. Sénizergues}

\author{Jonas Sénizergues}{Unaffiliated}{Jonas.Senizergues@free.fr}{https://orcid.org/0000-0002-9548-5260}{}

\authorrunning{J. Cohen, L. Pilard, F. Pirot, and  J. Sénizergues} 

\Copyright{Johanne Cohen, Laurence Pilard, François Pirot, and Jonas Sénizergues} 

\ccsdesc[500]{Theory of computation~Distributed algorithms}
\ccsdesc[300]{Theory of computation~Graph algorithms analysis}
\ccsdesc[300]{Mathematics of computing~Probabilistic algorithms}
\ccsdesc[100]{Mathematics of computing~Graph algorithms}

\keywords{Maximal Independent Set, Self-Stabilizing Algorithm} 

\newcommand{\motnouveau}[1]{\emph{#1}}
\newcommand{\Nat}{\mathbb{N}}

\newcommand{\ga}{\gamma_1}
\newcommand{\gb}{\gamma_2}
\newcommand{\E}{\mathcal{E}}
\newcommand{\R}{\mathcal{R}}

\begin{document}

\maketitle

\section*{Introduction}

The maximal independent set problem has received a lot of attention in different areas.  For instance, in wireless networks, maximum independent sets serve as a black box for communication (to collect or to broadcast information) (see~\cite{liu2017cooperative,gao2017novel}  for example).  In self-stabilizing distributed algorithms, this problem is also a fundamental tool to transform an algorithm from one model to another~\cite{grati2007, turau2006randomized}. 

 An \motnouveau{independent set} of a graph is a subset of its vertices such that no two of them form an edge in said graph. It is called \motnouveau{maximal} when it is maximal inclusion-wise (in which case it is also a minimal dominating set). 

The maximal independent set (MIS) problem has been extensively studied in parallel and distributed settings, following the seminal works of~\cite{Alon1986,linial1987, luby1986}.  Their idea is based on the fact that a node joins the ``MIS under construction'' $S$ according to its neighbors: node $v$ joins the set $S$ if it has no neighbor in $S$, and it leaves the set $S$ if at least one of its neighbors is in $S$.  Most algorithms in the literature, including ours, follow this approach.

The MIS problem has been extensively studied in the {\sc Local} model, \cite{ghaffari2016improved,rozhovn2020polylogarithmic,censor2020derandomizing} for instance  (a synchronous, message-passing model of distributed computing in which messages can be arbitrarily large) and in the {\sc Congest} model~\cite{peleg2000distributed} (a synchronous model where messages are $O(\log n)$ bits long). In the {\sc Local} model, Barenboim, Elkin, and Goldenberg~\cite{barenboim2018locally} focus on identified systems and gave a self-stabilizing algorithm producing a MIS within $O(\Delta + \log^{*} n)$ rounds.
Balliu, Brandt, Hirvonen, Olivetti, Rabie, and Suomela~\cite{balliu2019lower} prove that the previous algorithm~\cite{barenboim2018locally} is optimal for a wide range of parameters in the {\sc Local} model. In the {\sc Congest} model, Ghaffari, Grunau, and Rozhoň~\cite{ghaffari2021improved} prove that there exists a randomized distributed algorithm that computes a maximal independent set in $O(\log \Delta \cdot \log \log n + \log^{6} \log n)$ rounds with high probability.
 
Self-stabilizing algorithms for maximal independent set have been designed in various models (anonymous network~\cite{shukla1995observations,turau2006randomized,turau2019making}  or not~\cite{goddard2003self,ikeda2002space}).
To our knowledge, Shukla, Rosenkrantz, and Ravi~\cite{shukla1995observations} present the first algorithm designed for finding a MIS in a graph using the self-stabilization paradigm for anonymous networks.  Some other self-stabilizing works deal with this problem assuming identifiers: with a synchronous daemon~\cite{goddard2003self} or a distributed one~\cite{ikeda2002space}. These two works require $O(n^2)$ moves to converge. Turau~\cite{turau2007linear} improves these results to $O(n)$ moves under the distributed daemon. Recently, some works improved the results 
in the synchronous model. For non-anonymous networks,  Hedetniemi \cite{hedetniemi2020self} designed a self-stabilization algorithm for solving the problem related to dominating sets in graphs, in particular for a maximal independent set, which stabilizes in $O(n)$ synchronous rounds. Moreover,  for  anonymous networks, Turau~\cite{turau2019making} 
 design some Randomized self-stabilizing algorithms for maximal independent set  w.h.p. in $O(\log  n)$ rounds.
 See the survey~\cite{guellati2010survey} for more details on MIS self-stabilizing algorithms. 
 
Some variants of the maximal independent set problem have been investigated, such as the 1-maximal independent set problem \cite{tanaka2021self, SHI200477} or Maximal Distance-$k$ Independent Set \cite{benreguia2021selfstabilizing,johnen:hal-03138979}.
Tanaka, Sudo, Kakugawa, Masuzawa, and Datta~\cite{tanaka2021self} designed a silent self-stabilizing 1-MIS algorithm under the weakly-fair distributed daemon for any identified network in  $O(nD)$ rounds (where $D$ is the diameter of the graph).

This paper constitutes an extension of a brief announcement presented at SSS 2021~\cite{nous21}. We focus on the construction of a MIS, handling both transient and Byzantine faults. 
On one hand, transient faults can appear in the whole system, possibly impacting all nodes. However, these faults are not permanent, thus they stop at some point in the execution. 
Self-stabilization~\cite{dijkstra74} is the classical paradigm to handle transient faults. Starting from any arbitrary configuration, a self-stabilizing algorithm eventually resumes correct behavior without any external intervention.
On the other hand, (permanent) Byzantine faults~\cite{lamport82} are located on some faulty nodes, so the faults only occur from them. However, these faults can be permanent, \emph{i.e.}, they could never stop during the whole execution.

In a distributed system, multiple processes can be active at the same time.
The definition of a self-stabilizing algorithm is centered around the notion of \motnouveau{daemons}, which captures how the processes are scheduled at each time step of the execution.
Two major types of daemons are the \motnouveau{sequential} and the \motnouveau{distributed} ones. A \motnouveau{sequential daemon} only allows one process to be scheduled for a given time step, while a \motnouveau{distributed daemon} allows the execution of multiple processes at the same time.
Daemons can also be \motnouveau{fair} when they have to eventually schedule every process that is always activable, or \motnouveau{adversarial} when they are not fair.
As being distributed instead of sequential (or adversarial instead of fair) only allows for more possibilities of execution, it is harder to make an algorithm with the assumption of a distributed (resp. adversarial) daemon than with the assumption of a sequential (resp. fair) daemon.

We introduce the possibility that some nodes, that we will call \motnouveau{Byzantine} nodes, are not following the rules of the algorithm, and may change the values of their attributes at any time step. Here, if we do not work under the assumption of a fair daemon, one can easily see that we cannot guarantee the convergence of any algorithm, as the daemon could choose to always activate the same Byzantine node alone, again and again.

Under the assumption of a fair daemon, there is a natural way to express complexity, not in the number of moves performed by the processes, but in the number of \motnouveau{rounds}, where a round captures the idea that every process that wanted to be activated at the beginning of the round has either been activated or changed its mind. We give a self-stabilizing randomized algorithm that, in an 
\motnouveau{anonymous} network (which means that processes do not have unique identifiers to identify themselves) 
with Byzantine nodes under the assumption of a distributed fair daemon, finds a maximal independent set of a superset of the nodes at distance $3$ or more from Byzantine nodes. We show that the algorithm stabilizes in $O(\Delta n)$ rounds w.h.p., where $n$ is the size and $\Delta$ is the diameter of the underlying graph.

We present the model in Section~\ref{sec:model}. In Section~\ref{sec:Byzantine}, we give a self-stabilizing randomized algorithm with Byzantine nodes under the fair daemon, which converges in $O(\Delta n)$ rounds. Finally, in Section~\ref{sec:anonymous} we present a self-stabilizing randomized algorithm that finds a maximal independent set in an anonymous network, under the assumption of a distributed adversarial daemon. We show that our algorithm converges in $O(n^2)$ moves with high probability.

\section{Model}\label{sec:model}
A system consists of a set of processes where two adjacent processes can communicate with each other. The communication relation is represented by a graph $G = (V, E)$ where $V$ is the set of the processes (we will call \motnouveau{node} any element of $V$ from now on) and $E$ represents the neighborhood relation between them, \emph{i.e.}, $uv \in E$ when $u$ and $v$ are adjacent nodes.
By convention we write $|V|=n$ and $|E|=m$. If $u$ is a node, $N(u)=\quickset{v \in V | uv \in E}$ denotes the open neighborhood, and $N[u]=N(u) \cup \quickset{u}$ denotes the closed neighborhood. We note $deg(u) = |N(u)|$ and $\Delta = \max \quickset{deg(u)| u\in V}$. 

We assume the system to be \motnouveau{anonymous}, meaning that a node has no identifier.
We use the \motnouveau{state model}, which means that each node has a set of \motnouveau{local variables} which make up the \motnouveau{local state} of the node. 
A node can read its local variables and all the local variables of its neighbors, but can only rewrite its own local variables. A \motnouveau{configuration} is the value of the local states of all nodes in the system.
When $u$ is a node and $x$ a local variable, the \motnouveau{$x$-value} of $u$ in configuration $\gamma$ is the value $x_u^{\gamma}$.

An algorithm is a set of \motnouveau{rules}, where each rule is of the form $\langle guard \rangle \rightarrow \langle command \rangle$ and is parametrized by the node where it would be applied.
The \motnouveau{guard} is a predicate over the variables of the said node and its neighbors. The \motnouveau{command} is a sequence of actions that may change the values of the node's variables (but not those of its neighbors).
A rule is \motnouveau{enabled} on a node $u$ in a configuration $\gamma$ if the guard of the rule holds on $u$ in $\gamma$. A node is \motnouveau{activable} on a configuration $\gamma$ if at least one rule is enabled on $u$.
We call \motnouveau{move} any ordered pair $(u,r)$ where $u$ is a node and $r$ is a rule.
A move is said \motnouveau{possible} in a given configuration $\gamma$ if $r$ is enabled on $u$ in $\gamma$.

The \motnouveau{activation} of a rule on a node may only change the value of variables of that specific node, but multiple moves may be performed at the same time, as long as they act on different nodes. To capture this, we say that a set of moves $t$ is \motnouveau{valid} in a configuration $\gamma$ when it is non-empty, contains only possible moves of $\gamma$, and does not contain two moves concerning the same node.
Then, a \motnouveau{transition} is a triplet $(\gamma,t,\gamma')$ such that:
(i) $t$ is a valid set of moves of $\gamma$ and (ii) $\gamma'$ is a possible configuration after every node $u$ appearing in $t$ performs the code of the associated rule simultaneously, beginning in configuration $\gamma$.
We will write such a triplet as $\gamma \xrightarrow{t} \gamma'$. We will also write $\gamma \rightarrow \gamma'$ when there exists a transition from $\gamma$ to $\gamma'$.  $V(t)$ denotes the set of nodes that appear as the first member of a couple in $t$. 

We say that a rule $r$ is \motnouveau{executed} on a node $u$ in a transition $\gamma \xrightarrow{t} \gamma'$ (or equivalently that the move $(u,r)$ is \motnouveau{executed} in $\gamma \xrightarrow{t} \gamma'$) when the node $u$ has performed the rule $r$ in this transition, \textit{i.e.} when $(u,r) \in t$. In this case, we say that $u$ has been \motnouveau{activated} in that transition. 
Then, an \motnouveau{execution} is an alternate sequence of configurations and move sets $\gamma_0,t_1,\gamma_1 \dotsm t_i,\gamma_i, \dotsm $ where (i) the sequence either is infinite or finishes by a configuration and (ii) for all $i \in \Nat$ such that it is defined, $(\gamma_i,t_{i+1},\gamma_{i+1})$ is a transition.
We will write such an execution as $\gamma_0\xrightarrow{t_1}\gamma_1 \dotsm \xrightarrow{t_i}\gamma_i \dotsm $
When the execution is finite, the last element of the sequence is the \motnouveau{last configuration} of the execution.
An execution is \motnouveau{maximal} if it is infinite, or it is finite and no node is activable in the last configuration. It is called \motnouveau{partial} otherwise.
We say that a configuration $\gamma'$ is \motnouveau{reachable} from a configuration $\gamma$ if there exists an execution starting in configuration $\gamma$ that leads to configuration $\gamma'$. We say that a configuration is \emph{stable} if no node is activable in that configuration.
 
The \motnouveau{daemon} is the adversary that chooses, from a given configuration, which nodes to activate in the next transition. Two daemons are used in this work: 
the \motnouveau{adversarial distributed daemon} that allows all possible executions and  the \motnouveau{fair distributed daemon} 
that only allows executions where nodes cannot be continuously activable without being eventually activated.

Given a specification and $\mathcal{L}$ the associated set of \motnouveau{legitimate configuration}, \emph{i.e.}, the set of the configurations that verify the specification, a probabilistic algorithm is \motnouveau{self-stabilizing} when these properties are true: (\motnouveau{correctness}) every configuration of an execution starting by a configuration of $\mathcal{L}$ is in $\mathcal{L}$ and (\motnouveau{convergence}) from any configuration, whatever the strategy of the daemon, the resulting execution eventually reaches a configuration in $\mathcal{L}$ with probability~$1$.

The time complexity of an algorithm that assumes the fair distributed daemon is given as a number of \motnouveau{rounds}, a concept introduced by Dolev, Israeli, and Moran~\cite{doismo97}. It was reworded by Cournier, Devismes, and Villain~\cite{codevi2006} to take into account activable nodes, from whom we quote the following two definitions:

\begin{definition}
\label{pif_disabling_action_df}
We consider that a node $u$ executes a \motnouveau{disabling action} in the transition $\ga \to \gb$ if $u$ (i) is activable in $\ga$, (ii) does not execute any rule in $\ga \to \gb$, and (iii) is not activable in $\gb$.
\end{definition}

The disabling action represents the situation where at least one neighbor of $u$ changes its local state in $\ga \to \gb$, and this change effectively makes the guard of all rules of $u$ false in $\gb$. The time complexity is then computed by capturing the speed of the slowest node in any execution through the round definition~\cite{doismo97}.

\begin{definition}
\label{pif_round_df}
Given an execution $\E$, the first \emph{round} of $\E$ (let us call it $\R_1$) is the minimal prefix of $\E$ containing the execution of one action (the execution of a rule or a \emph{disabling action}) of every activable node from the initial configuration.
Let $\E'$ be the suffix of $\E$ such that $\E = \R_1 \E'$. The second round of $\E$ is the first round of $\E'$, and so on.''
\end{definition}
 
Observe that Definition~\ref{pif_round_df} is equivalent to Definition~\ref{my_round_df}, which is simpler in the sense that it does not refer back to the set of activable nodes from the initial configuration of the round.

\begin{definition}
\label{my_round_df}
Let $\E$ be an execution. A \emph{round} is a sequence of consecutive transitions in $\E$.
The first round begins at the beginning of $\E$; successive rounds begin immediately after the previous round has ended.
The current round ends once every node $u\in V$ satisfies at least one of the following two properties:
(i) $u$ has been activated in at least one transition during the current round or (ii) $u$ has been non-activable in at least one configuration during the current round.
\end{definition}

Our first algorithm is to be executed in the presence of \motnouveau{Byzantine} nodes; that is, there is a subset $B \subseteq V$ of adversarial nodes that are not bound by the algorithm. Byzantine nodes are always activable.  An activated Byzantine node is free to update its local variables arbitrarily. Finally, observe that, in the presence of Byzantine nodes, all maximal executions are infinite.
We denote by $d(u,B)$  the minimal (graph) distance between node $u$ and a Byzantine node, and we define for $i\in \mathbb{N}$:
$V_i = \quickset{u \in V | d(u,B)>i }$. 
Note that $V_0$ is exactly the set of non-Byzantine nodes, and that $V_{i+1}$ is exactly the set of nodes of $V_i$ whose neighbors are all in $V_i$.

When $i,j$ are integers, we use the standard mathematical notation for the integer segments:
 $\segment{i}{j}$ is the set of integers that are greater than or equal to $i$ and smaller than or equal to $j$ (\emph{i.e.} $\segment{i}{j}=[i,j]\cap \mathbb{Z} = \lbrace i, \dotsm, j \rbrace$).
$Rand(x)$ with $x \in [0,1]$ represents the random function that outputs $1$ with probability $x$, and $0$ otherwise.

\section{With Byzantine Nodes under the Fair Daemon}
\label{sec:Byzantine}

\subsection{The algorithm}

The algorithm builds a maximal independent set represented by a local variable~$s$. The  approach of the state of the art is the following: when two nodes are candidates to be in the independent set, then a local election  decides who will remain in the independent set. To perform a local election,  the standard technique is to compare the identifiers of nodes. Unfortunately,  this mechanism is not robust to the presence of Byzantine nodes.

Keeping with the approach outlined above, when a node $u$ observes that its neighbors are not in (or trying to be in) the independent set, the non-Byzantine node decides to join it with a certain probability.  The randomization helps to reduce the impact of Byzantine nodes. The choice of probability should reduce the impact of Byzantine nodes while maintaining the efficiency of the algorithm.

\begin{algo}
Any node $u$ has two local variables $s_u \in \quickset{\bot,\top}$ and $x_u \in \mathbb{N}$ and may make a move according to one of the following rules: \\
\textbf{(Refresh)} $x_u \not =  |N(u)| \rightarrow x_u := |N(u)| \quad (=deg(u))$ \\
\textbf{(Candidacy?)} $(x_u =  |N(u)|) \wedge (s_u=\bot)\wedge(\forall v \in N(u), s_v=\bot) \rightarrow$ \\ if $Rand(\frac{1}{1+\max(\quickset{x_v | v \in N[u]})})=1$ then $s_u := \top$\\
\textbf{(Withdrawal)} $(x_u = |N(u)|) \wedge (s_u=\top)\wedge(\exists v \in N(u), s_v =\top) \rightarrow s_u := \bot$
\end{algo}

Observe that, since we assume an anonymous setting, the only way to break symmetry is randomization.
The values of the probabilities for changing local variable values must be carefully chosen to contain the impact of the Byzantine node. 

A node joins the MIS with a probability $\frac{1}{1+\max(\quickset{x_v | v \in N[u]})}$. 
The idea to ask the neighbors about their own number of neighbors (through the use of the $x$ variable) to choose the probability of a candidacy comes from the mathematical property $\forall k \in \mathbb{N}, (1-\frac{1}{k+1})^{k} > e^{-1}$, which will allow us to have a good lower bound for the probability of the event ``some node made a successful candidacy, but none of its neighbors did''.

\subsection{Specification}

Since Byzantine nodes are not bound to follow the rules, we cannot hope for a correct solution in the entire graph. But we want to find a solution that works when we are far enough from the Byzantine nodes. 
One could think about a fixed containment radius around Byzantine nodes, but as we will see later on, this is not as simple.

Let us define, for any configuration $\gamma$, the following set of nodes that represents the already built independent set: 
$$I_{\gamma} = \quickset{u \in V_1 | (s_{u}^{\gamma}=\top) \wedge \forall v \in N(u), s_{v}^{\gamma}=\bot}$$
~\\[-1,4em]
We say that a node is \emph{locally alone} if it is a candidate to be in the independent set (\emph{i.e.}, its $s$-value is $\top$) while none of its neighbors are.
In configuration $\gamma$, $I_{\gamma}$ is the set of all locally alone nodes of $V_1$.   

\begin{definition}
A configuration is said \emph{legitimate} when $I_\gamma$ is a maximal independent set of $V_2 \cup I_\gamma$.
\end{definition}

\subsubsection{An example}

\newcommand{\topo}[4]{
 \begin{tikzpicture}[scale=0.55]
\tikzstyle{vertex}=[fill=white, draw=black, shape=circle]
\tikzstyle{edge}=[-]

                 \node[style=vertex,shape =rectangle]   (b) at (-12, 0) {$b$};       
                 \node [style=vertex] (c) at (-10, 0) {$v_{1}$}; 
 	         \node [style=vertex] (d) at (-8, 0) {$v_{2}$};   		
	         \node [style=vertex] (e) at (-6, 0) {$v_{3}$};  
 
   		\draw [style=edge] (b) to (c);
 		\draw [style=edge] (c) to (d);
		 \draw [style=edge] (e) to (d);

		\draw (b.north) node[above]{#1} ;
 	         \draw (c.north) node[above]{#2} ;
 	         \draw (d.north) node[above]{#3} ;
 		\draw (e.north) node[above]{#4} ;
\end{tikzpicture}
}

Figure~\ref{fig:byzantin} gives an example of an execution of the algorithm. 
Figure~\ref{fig:byzantin}a depicts a network in a given configuration.  The symbol drawn above the node represents the local variable $s$. Each local variable $x$ contains the degree of its associated node. 
Byzantine node is shown with a square.  

In the initial configuration, nodes  $v_{1}$ and $v_{2}$ are in the independent set, and then are activable for \textbf{Withdrawal}.  In the first step, the daemon activates $v_{1}$ (\textbf{Withdrawal}) and  $v_{2}$  (\textbf{Withdrawal}) leading to configuration $\gamma_1$ (Fig. \ref{fig:byzantin}b).
 In the second step, the daemon activates $v_{1}$ (\textbf{Candidacy?}). Node $v_{1}$ randomly decides whether to set $s_{v_{1}}:=\top $  leading to configuration $\gamma_2$ (Fig. \ref{fig:byzantin}c), or $s_{v_{1}}:=\bot $ leading to configuration $\gamma_1$ (Fig. \ref{fig:byzantin}b).
Assume that $v_{1}$  chooses $s_{v_{1}}:=\top $. At this moment, node $v_{1}$ is ``locally alone'' in the independent set.  In the third step, the daemon activates $b$ and  $b$ makes a Byzantine move, setting $s_{b}:=\top $, leading to configuration $\gamma_3$ (Fig. \ref{fig:byzantin}d).

In the fourth step, the daemon activates $v_{1}$   (\textbf{Withdrawal}) and $b$ that sets $s_{b}:=\bot $.
  The configuration is now the same as the first configuration (Fig. \ref{fig:byzantin}b). The daemon is assumed to be fair, 
 so nodes $v_{2}$ and $v_{3}$ need to be activated before the execution can be called an infinite loop. These activations will prevent the node $v_1$ from alternating forever between in and out of the independent set, while the rest of the system remains out of it. 
  In the fifth step, the daemon activates $v_{1}$ (\textbf{Candidacy?}), $v_{2}$ (\textbf{Candidacy?}) and  $v_{3}$ (\textbf{Candidacy?}). They randomly decide to change their local variable $s$.
  Assume that $v_{1}$, $v_{2}$ and $v_{3}$ choose $s_{v_{1}}:=\top $, $s_{v_{2}}:=\bot $, and $s_{v_{3}}:=\top $, leading to configuration $\gamma_5$ (Fig. \ref{fig:byzantin}e). At this moment, node $v_{3}$ is ``locally alone'' in the independent set.  As $v_3$ is far enough from the Byzantine node, it will remain in the independent set whatever $b$ does.

\begin{figure}[h]
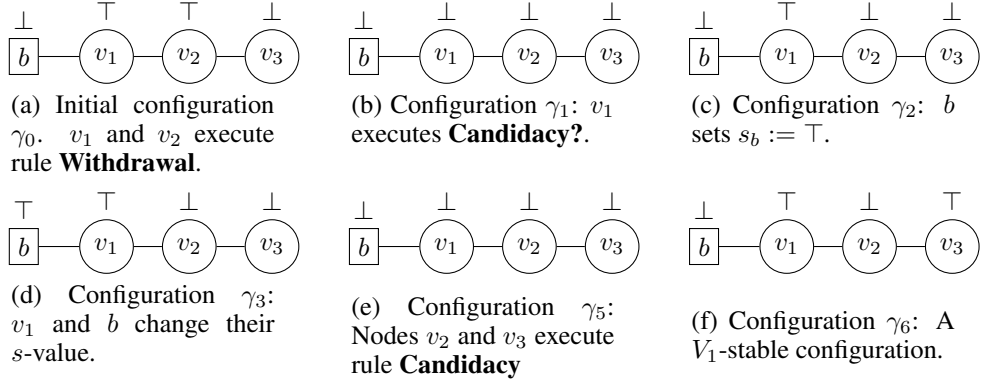

\centering
\begin{tabular}{p{4.1cm}p{4.1cm}p{4.1cm}}
 \topo{$\bot$}{$\top$}{$\top$}{$\bot$}  &  \topo{$\bot$}{$\bot$}{$\bot$}{$\bot$}  &  \topo{$\bot$}{$\top$}{$\bot$}{$\bot$}  \\
~\begin{minipage}[h]{3.5cm}
(a) Initial configuration $\gamma_0$. $v_{1}$ and  $v_{2}$  execute rule \textbf{Withdrawal}. 
\end{minipage}
  & 
~\begin{minipage}[h]{3.5cm}
(b) Configuration $\gamma_1$: $v_{1}$  executes \textbf{Candidacy?}.  \\
\end{minipage}
  &
~\begin{minipage}[h]{3.5cm}
   (c) Configuration  $\gamma_2$: $b$ sets  $s_{b}:=\top $.\\
\end{minipage}
   \\[1,3em]
 \topo{$\top$}{$\top$}{$\bot$}{$\bot$}   &
\topo{$\bot$}{$\bot$}{$\bot$}{$\bot$}    &   \topo{$\bot$}{$\top$}{$\bot$}{$\top$}  
   \\
~\begin{minipage}[h]{3.5cm}
(d) Configuration  $\gamma_3$: $v_{1}$ and $b$ change their $s$-value.\\
\end{minipage}
  & 
~\begin{minipage}[h]{3.5cm}
(e) Configuration  $\gamma_5$: Nodes $v_2$ and $v_3$ execute rule \textbf{Candidacy}
\end{minipage}
 &
 ~\begin{minipage}[h]{3.5cm}
(f) Configuration  $\gamma_6$: A $V_1$-stable configuration.
\end{minipage}
\end{tabular}
\caption{Execution}
\label{fig:byzantin}

\end{figure}

\subsubsection{About the specification:}

Our goal is to design an algorithm that builds a maximal independent set of the subgraph induced by a set of nodes where nodes ``too close'' to Byzantine nodes have been removed. 
The question now is to define what ``too close'' means. One could think about a fixed containment radius, only excluding nodes at distance at most 1 from Byzantine nodes. This set of nodes has been previously defined as  $V_1$. Indeed, in Figure~\ref{ex:spec}.(a),  $v_1$ and $v_2$ belong to $V_1$ and their local view of the system is correct, then they have no reason to change their states. Moreover, Byzantine nodes are too far away to change that: whatever the value of the state of $v_0$, the view of $v_1$ remains correct. 
Thus, a containment radius of 1 could seem correct. However, in Figure~\ref{ex:spec}.(b), if the Byzantine node does not make any move, then $v_0$ remains in the MIS while $v_1$ remains out of it. Thus, in this example, if we only consider nodes in $V_1$, the $\top$-valued nodes of $V_1$ are not a MIS of $V_1$. In fact, this situation can be extended to any $V_i$, the idea being that a node of $V_i$ might be in a stable state without a neighbor of $V_i$ having value $\top$ if it has a neighbor of $V_{k-1}\setminus V_k$ with value $\top$ (see Figure~\ref{ex:spec}.(c/d) for examples for $k=2$ and $k=3$). To circumvent this problem, we consider a set of nodes defined from a fixed containment radius to which we add locally alone neighboring nodes. The smallest containment radius that works with this approach is 3 (which corresponds to set $V_2$).  Note that it depends on the current configuration and not only on the underlying graph. 

\subsubsection{About the choice of probability to join the MIS} 

We could have gone with the same probability for every node, but that comes with the cost of making the algorithm very sensitive to the connectivity of the underlying graph. As we rely for convergence on the event where a node is candidate alone (\emph{i.e.} switch to $\top$ without other nodes doing the same in its neighborhood), the probability of progress in a given number of rounds would then be exponentially decreasing with the degrees of the graph.

We could have gone with something depending only on the degree of the node where the rule is applied. While it could have been an overall improvement over the uniform version above, the minoration of the probability of progress that can be made with a local scope is no better. We cannot exclude that a finer analysis would lead to a better overall improvement, but it would require dealing with far more complex math. On a smaller scale, we can also note that this choice would introduce a bias toward small-degree nodes, while we might not want that (depending on the application).

Then, we have chosen the version where nodes take into account their degree and what they know of the degree of their neighbors. On one hand, the first concern you could raise here would be the potential sabotage by Byzantine nodes. Here is the intuition for why this cannot be a problem here. If a node $u$ is at distance at least 2 from any Byzantine node and if $u$ is ``locally alone'' in the independent set, then whatever the Byzantine nodes do, $u$ will forever remain in the independent set. To maximize the harm done, the Byzantine nodes have to indirectly prevent such a node from joining the independent set in the first place. To do so it has to maximize the probability of its neighbors to be candidate to the independent set. But Byzantine nodes cannot lie efficiently in that direction, as the probability is upper-bounded by the degree values of both the node and its non-Byzantine neighbors.
On the other hand, this choice allows us to address the problem that we had with the previous solution. Here, we can indeed frame the probability of being a candidate alone between two constant bounds with a simple local analysis. Thus, we can ensure that the convergence speed does not depend on the connectivity of the underlying graph. Again, on a smaller scale, we also greatly reduce the bias toward small-degree nodes compared to the previous option.

\newcommand{\topoDeux}[3]{
 \begin{tikzpicture}[scale=0.55]

\tikzstyle{vertex}=[fill=white, draw=black, shape=circle]
\tikzstyle{edge}=[-]

                 \node[style=vertex,shape =rectangle]   (b) at (-12, 0) {$b$};       
                 \node [style=vertex] (c) at (-10, 0) {$v_{0}$}; 
 	         \node [style=vertex] (d) at (-8, 0) {$v_{1}$};   		
 
   		\draw [style=edge] (b) to (c);
 		\draw [style=edge] (c) to (d);

		\draw (b.north) node[above]{#1} ;
 	         \draw (c.north) node[above]{#2} ;
 	         \draw (d.north) node[above]{#3} ;
 \end{tikzpicture}
}

\newcommand{\topoTrois}[4]{
 \begin{tikzpicture}[scale=0.55]

\tikzstyle{vertex}=[fill=white, draw=black, shape=circle]
\tikzstyle{edge}=[-]

                 \node[style=vertex,shape =rectangle]   (b) at (-12, 0) {$b$};       
                 \node [style=vertex] (c) at (-10, 0) {$v_{0}$}; 
 	         \node [style=vertex] (d) at (-8, 0) {$v_{1}$};   		
	         \node [style=vertex] (e) at (-6, 0) {$v_{2}$};  

   		\draw [style=edge] (b) to (c);
 		\draw [style=edge] (c) to (d);
		 \draw [style=edge] (e) to (d);

		\draw (b.north) node[above]{#1} ;
 	         \draw (c.north) node[above]{#2} ;
 	         \draw (d.north) node[above]{#3} ;
	         \draw (e.north) node[above]{#4} ;

 \end{tikzpicture}
}

\newcommand{\topoCinq}[6]{
 \begin{tikzpicture}[scale=0.55]

\tikzstyle{vertex}=[fill=white, draw=black, shape=circle]
\tikzstyle{edge}=[-]

                 \node[style=vertex,shape =rectangle]   (b) at (-12, 0) {$b$};       
                 \node [style=vertex] (c) at (-10, 0) {$v_{0}$}; 
 	         \node [style=vertex] (d) at (-8, 0) {$v_{1}$};   		
	         \node [style=vertex] (e) at (-6, 0) {$v_{2}$};  
 	         \node [style=vertex] (f) at (-4, 0) {$v_{3}$};  
 	         \node [style=vertex] (g) at (-2, 0) {$v_{4}$};  

   		\draw [style=edge] (b) to (c);
 		\draw [style=edge] (c) to (d);
		 \draw [style=edge] (e) to (d);
		 \draw [style=edge] (e) to (f);
		 \draw [style=edge] (f) to (g);

		\draw (b.north) node[above]{#1} ;
 	         \draw (c.north) node[above]{#2} ;
 	         \draw (d.north) node[above]{#3} ;
	         \draw (e.north) node[above]{#4} ;
 		 \draw (f.north) node[above]{#5} ;
 		 \draw (g.north) node[above]{#6} ;

 \end{tikzpicture}
}
\newcommand{\topoSix}[7]{
 \begin{tikzpicture}[scale=0.55]

\tikzstyle{vertex}=[fill=white, draw=black, shape=circle]
\tikzstyle{edge}=[-]

                 \node[style=vertex,shape =rectangle]   (b) at (-12, 0) {$b$};       
                 \node [style=vertex] (c) at (-10, 0) {$v_{0}$}; 
 	         \node [style=vertex] (d) at (-8, 0) {$v_{1}$};   		
	         \node [style=vertex] (e) at (-6, 0) {$v_{2}$};  
 	         \node [style=vertex] (f) at (-4, 0) {$v_{3}$};  
 	         \node [style=vertex] (g) at (-2, 0) {$v_{4}$};  
  	         \node [style=vertex] (h) at (0, 0) {$v_{5}$};  

   		\draw [style=edge] (b) to (c);
 		\draw [style=edge] (c) to (d);
		 \draw [style=edge] (e) to (d);
		 \draw [style=edge] (e) to (f);
		 \draw [style=edge] (f) to (g);
 		 \draw [style=edge] (g) to (h);

		\draw (b.north) node[above]{#1} ;
 	         \draw (c.north) node[above]{#2} ;
 	         \draw (d.north) node[above]{#3} ;
	         \draw (e.north) node[above]{#4} ;
 		 \draw (f.north) node[above]{#5} ;
 		 \draw (g.north) node[above]{#6} ;
 		 \draw (h.north) node[above]{#7} ;

 \end{tikzpicture}
}

\begin{figure}[ht]
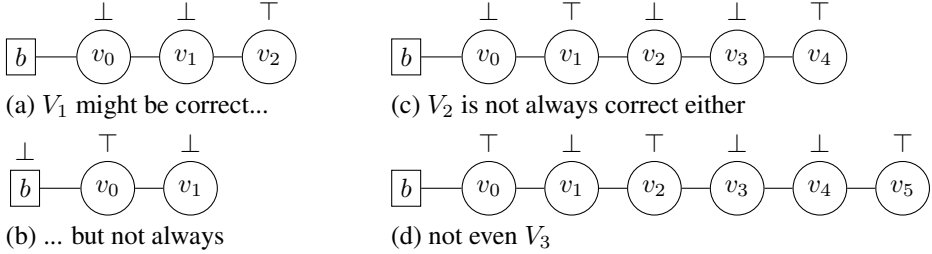

\begin{tabular}{p{4.7cm}p{6.1cm}}
\topoTrois{}{$\bot$}{$\bot$}{$\top$} & 
\topoCinq{}{$\bot$}{$\top$}{$\bot$}{$\bot$}{$\top$}\\
(a) $V_1$ might be correct...  & (c) $V_2$ is not always correct either
\end{tabular}
~\\
\begin{tabular}{p{4.7cm}p{6.1cm}}

\topoDeux{$\bot$}{$\top$}{$\bot$} & 
\topoSix{}{$\top$}{$\bot$}{$\top$}{$\bot$}{$\bot$}{$\top$}\\
(b) ... but not always & (d) not even $V_3$
\end{tabular}
\caption{What is the right containment radius?}
\label{ex:spec}
\end{figure}

\subsection{The proof}

We say that in a configuration $\gamma$, a node $u$ is \motnouveau{degree-stabilized} if rule \textbf{Refresh} is not enabled on it. The configuration  $\gamma$ is then said to be \motnouveau{degree-stabilized} if every non-Byzantine node is degree-stabilized.

Observe the following two facts about degree-stabilization.

\begin{lemma}
Any reachable configuration from a degree-stabilized configuration is degree-stabilized.
\end{lemma}

\begin{proof}
No rule can change the $x$ value of a degree-stabilized non-Byzantine node.
\end{proof}

\begin{lemma} \label{bz-lem-stab-1round}
From any configuration $\gamma$, the configuration $\gamma'$ after one round is degree-stabilized.
\end{lemma}

\begin{proof}
Let $u$ be a non-Byzantine node.

If $x_{u}^{\gamma}=deg(u)$, no activation of rule can change that thus $x_{u}^{\gamma'}=deg(u)$.

If $x_{u}^{\gamma}\not = deg(u)$, then $u$ is activable in $\gamma$ and remains so until it is activated. Rule \textbf{Refresh} is then executed on $u$ in the first round, and since no rule can change the value of $x_u$ afterward, we have $x_{u}^{\gamma'} = deg(u)$.
\end{proof}

All locally alone nodes in $V_{1}$ (\emph{i.e.} nodes of $I_\gamma$) remain locally alone during the whole execution.

\begin{lemma} \label{bz-lem-I-inc}

If $\gamma \rightarrow \gamma'$, $I_{\gamma} \subseteq I_{\gamma'}$.
\end{lemma}

\begin{proof}
Let's consider $u\in I_{\gamma}$. The only rules that may be enabled on $u$ in $\gamma$ is \textbf{Refresh} since \textbf{Candidacy?} can't be executed on $u$ because $s_u^{\gamma} = \top$, and \textbf{Withdrawal} can't be executed on $u$ because $\forall v\in N(u), s_u^{\gamma} = \bot$. We have then $s_{u}^{\gamma'} = \top$.

Now let's consider $v \in N(u)$. By definition of $I_{\gamma}$, $s_{v}^{\gamma}=\bot$, and $v$ has $u$ as a neighbor that has value $\top$ in $\gamma$. Then the only rule that may be enabled on $v$ in $\gamma$ is \textbf{Refresh}, and we have $s_{v}^{\gamma'}=\bot$.

Thus, $u \in I_{\gamma'}$.
\end{proof}

We now focus on the progression properties of enabled rules after one round.  We start with the
\textbf{Withdrawal} rule.  When the  \textbf{Withdrawal} rule is enabled on a node $u \in V_{1}$,  the conflict is solved by either making $u$ locally alone or setting $s_{u} = \bot $.

\begin{lemma} \label{bz-lem-withdrawal}
If $\gamma$ is a degree-stabilized configuration and if \textbf{Withdrawal} is enabled on $u\in V_1$, then after one round either \textbf{Withdrawal} has been executed on $u$ or in the resulting configuration $\gamma'$ we have $u \in I_{\gamma'}$.
\end{lemma}

\begin{proof}
Since $\gamma$ is degree-stabilized, no \textbf{Refresh} move can be executed in any future transition. Then, since $s_u^{\gamma}=\top$, only \textbf{Withdrawal} can be executed on $u$ or any of its non-Byzantine neighbors until $u$ has been activated. Since $u \in V_1$, it is in fact true for every neighbor of $u$.

\noindent Since $u$ is activable in $\gamma$, we have two cases:
(i) If $u$ has performed a \textbf{Withdrawal} move in the next round, there is nothing left to prove.
(ii) If it is not the case, by fairness hypothesis $u$ must have been unactivated, which means that each neighbor $v\in N(u)$ that had $s$-value $\top$ in $\gamma$ has been activated. By the above, they must have performed a \textbf{Withdrawal} move that changed their $s$-value to $\bot$. Also, $u$ is supposed not to have performed any rule, so it keeps $s$-value $\top$ in the whole round: its neighbors -that are non-Byzantine since $u\in V_1$- cannot perform any \textbf{Candidacy} move. As such, in the configuration $\gamma'$ at the end of the round, every neighbor of $u$ has $s$-value $\bot$, and $u$ has $s$-value $\top$. Since $u\in V_1$ that means that $u \in I_{\gamma'}$. 
\end{proof}

When the \textbf{Candidacy?} rule is enabled on a node $u \in V_{1}$, then \textbf{Candidacy?} is executed on $v \in N[u]$ within one round.

\begin{lemma} \label{bz-lem-candidacyworks}
If in $\gamma$ degree-stabilized we have the \textbf{Candidacy?} rule enabled on $u \in V_1$ (\emph{i.e.}, $s_{u}^{\gamma} = \bot$ and $\forall v \in N(u), s_{v}^{\gamma} = \bot$) then after one round \textbf{Candidacy?} have either been executed on $u$, or on at least one neighbor of $u$.
\end{lemma}

\begin{proof}
Since $\gamma$ is degree-stabilized, no \textbf{Refresh} move can be executed in any future transition. Then, until $u$ or one of its neighbors has been activated, only \textbf{Candidacy?} can be executed on them since it's the only rule that can be activated on a node with $s$-value $\bot$.

\noindent Since $u$ is activable in $\gamma$, by fairness, we have two cases:

(i) If $u$ is activated before the end of the round, the only rule that it could have performed for its first activation is the \textbf{Candidacy?} rule since its $s$-value in $\gamma$ is $\bot$ and the configuration is supposed degree-stabilized.

(ii) If not, it has been unactivated, which means that at least one neighbor $v\in N(u)$ has been activated. As $u\in V_1$, $v$ cannot be Byzantine. The only rule that it could have performed for its first activation is the \textbf{Candidacy?} rule since its $s$-value in $\gamma$ is $\bot$ and the configuration is supposed degree-stabilized. 
\end{proof}

\begin{fact}\label{fact:geinve}
$\forall k \in \mathbb{N}, (1-\frac{1}{k+1})^{k} > e^{-1}$
\end{fact}

\begin{proof}
For $k=0$ it is true ($(1-\frac{1}{0+1})^{0} = 1 > \frac{1}{e}$). Suppose now that $k \geq 1$.

A basic inequality about $\ln$ is that $\forall x\in \mathbb{R}^{+*}, \ln(x) \geq 1-\frac{1}{x}$, with equality only when $x=1$.
For $x=\frac{k}{k+1}$, it gives us $\ln \left(\frac{k}{k+1} \right) \geq 1 - \frac{k+1}{k}= -\frac{1}{k}$, with equality only when $\frac{k}{k+1} = 1$. Since $\frac{k}{k+1}$ cannot have value $1$ for any value of $k$, we have $\ln \left(\frac{k}{k+1} \right) > -\frac{1}{k}$.

Then, by multiplying by $k$ on each side and taking the exponential :
$$\left(1 - \frac{1}{k+1} \right)^k = e^{k \ln \left(\frac{k}{k+1} \right)} > e^{-1}$$
\end{proof}

\noindent If node $u\in V_{1}$ executes the \textbf{Candidacy?} rule, then $u$ becomes a locally alone node with a certain probability in the next configuration. So it implies that the set $I$~grows.

\begin{lemma}\label{bz-lem-main}
If $\gamma$ is a degree-stabilized configuration, and in the next transition rule \textbf{Candidacy?} is executed on a node $u\in V_1$, there is probability at least $\frac{1}{e(\Delta+1)}$ that in the next configuration $\gamma'$, $s_u^{\gamma'}=\top$ and $\forall v\in N(u), s_v^{\gamma'}=\bot$.
\end{lemma}

\begin{proof}
For any node $y$, we write $\varphi(y) =  \frac{1}{1+max(\quickset{deg(v) | v \in N[y]})}$.

Since \textbf{Candidacy?} can be executed on $u$ in $\gamma$, we know that $\forall v \in N[u], s_{v}^{\gamma} = \bot$. Since $\gamma$ is degree-stabilized and no neighbor of $u$ can be Byzantine by definition of $V_1$, we have $\forall v \in N[u], x_{v}^{\gamma} = deg(v)$. The probability of $s_u^{\gamma'}=\top$ knowing the rule has been executed is then $\varphi(u)$.

Then, for a given $v \in N(u)$, node $v$ is not Byzantine since $u \in V_1$, and the probability that $s_v^{\gamma'} = \top$ is either $0$ (if \textbf{Candidacy?} has not been executed on $v$ in the transition) or $\frac{1}{1+\max(\quickset{x_w^{\gamma} | w \in N[v]})}$. Since $deg(u)= x_u \leq \max(\quickset{x_w^{\gamma} | w \in N[v]})$, that probability is then at most $\frac{1}{1+deg(u)}$.
Thus (since those events are independent), the probability for $u$ to be candidate in $\gamma'$ without a candidate neighbor is at least:
$$p = \varphi(u)\prod_{v \in N(u)} \left(1-\frac{1}{1+deg(u)} \right) \geq \frac{1}{\Delta+1} \left(1-\frac{1}{1+deg(u)} \right)^{deg(u)}$$
\noindent Then, as $\forall k \in \mathbb{N}, (1-\frac{1}{k+1})^{k} > e^{-1}$ (see Fact~\ref{fact:geinve}), $p > \frac{1}{\Delta+1} \times \frac{1}{e}$
 and the lemma holds. 
\end{proof}

\begin{lemma} \label{bz-lem-3way}
If $\gamma$ is a degree-stabilized configuration such that $I_{\gamma}$ is not a maximal independent set of $V_2 \cup I_{\gamma}$, then after at most one round one of the following events happens:\\
\begin{enumerate}
\item Rule \textbf{Candidacy?} is executed on a node of $V_1$
\item A configuration $\gamma'$ such that $I_{\gamma} \subsetneq I_{\gamma'}$ is reached.
\item A configuration $\gamma'$ such that rule \textbf{Candidacy?} is enabled on a node of $V_2$ in $\gamma'$ is reached.
\end{enumerate}
\end{lemma}

\begin{proof}
Suppose $\gamma$ is a degree-stabilized configuration such that $I_{\gamma}\cup V_2$ is not a maximal independent set of $V_2$. As $\gamma$ is supposed to be degree-stabilized, we will only consider the possibility of moves that are not \textbf{Refresh} moves.
(a) If \textbf{Candidacy?} is enabled on a node of $V_2$ in $\gamma$, Condition~$3$ holds.
(b) If it is not the case, then there exists $u \in V_2$ that has at least one neighbor $v \in V_1$ such that $s_{u}^{\gamma} = s_{v}^{\gamma} = \top$  (otherwise $I_{\gamma}$ would be a maximal independent set of $V_2 \cup I_{\gamma}$).

In the first case, there is nothing left to prove. In the second case, \textbf{Withdrawal} is enabled on $u \in V_2$ in $\gamma$ and from Lemma~\ref{bz-lem-withdrawal}, we have two possible cases: (i) if $\gamma'$ is the configuration after one round, $u \in I_{\gamma'}$ and Condition~$2$ holds ; (ii) $u$ performs a \textbf{Withdrawal} move before the end of the round.

In the first case there is nothing left to prove. Suppose now that we are in the second case and that $u$ is activated only once before the end of the round, without loss of generality since Condition~$1$ would hold otherwise, as its second activation would be a \textbf{Candidacy?} move and $u \in V_1$. 
Then: \\
\begin{itemize} 
\item If within a round a configuration $\gamma'$ is reached where a node $w \in N[u]$ is such that $s_w^{\gamma'} = \top$ and $w$ is not activable we have: $w \not\in I_{\gamma}$ (as $u$ is $\top$-valued in $\gamma$) which gives $I_{\gamma} \subsetneq I_{\gamma} \cup \quickset{w} \subseteq I_{\gamma'}$ by Lemma~\ref{bz-lem-I-inc}, thus Condition~$2$ holds.

\item If we suppose that no such event happens until the end of the round in configuration $\gamma'$, we are in either of those cases:
(i)  Every neighbor of $u$ has value $\bot$ in $\gamma'$ and $s_u^{\gamma'} = \bot$ (as we would be in the previous case if it were $\top$), thus \textbf{Candidacy?} is enabled on $u$ in $\gamma'$, and Condition~$3$ holds.
(ii)  A \textbf{Candidacy?} has been performed within the round on a neighbor of $u$ and Condition~$1$ holds.
\end{itemize}
~\\
Thus, in every possible case, one of the three conditions holds. 
\end{proof}
~\\
\noindent  Every 2 rounds, the set of locally alone\,nodes strictly\,grows with\,some probability. 
 
\begin{lemma} \label{bz-lem-advance}
If $\gamma$ is degree-stabilized, with $I_{\gamma}$ not being a maximal independent set of $V_2 \cup I_{\gamma}$, then after two rounds the probability for the new configuration $\gamma'$ to be such that $I_{\gamma} \subsetneq I_{\gamma'}$ is at least $\frac{1}{(\Delta+1)e}$.
\end{lemma}

\begin{proof}
From Lemma~\ref{bz-lem-3way}, we have three possibilities after one round.\\
\begin{itemize}
\item If we are in Case~$1$, let us denote by $\gamma''$ the resulting configuration after the transition where the said \textbf{Candidacy?} move have been executed on $u\in V_1$. Then using Lemma~\ref{bz-lem-main}, we have $u \in I_{\gamma''}$ with probability at least $\frac{1}{(\Delta+1)e}$. Since we have $u\not\in I_{\gamma}$  (if $u$ was in $I_{\gamma}$, \textbf{Candidacy?} could not have been executed on $u$ after configuration $\gamma$) and by Lemma~\ref{bz-lem-I-inc}, we have then $I_{\gamma} \subsetneq I_{\gamma'}$ with probability at least $\frac{1}{(\Delta+1)e}$.
\item If we are in Case~$2$, there is nothing left to prove.
\item If we are in Case~$3$, let us denote by $\gamma''$ the first configuration where \textbf{Candidacy?} is enabled on some $u\in V_2$. Then, using Lemma~\ref{bz-lem-candidacyworks}, after at most one more round \textbf{Candidacy?} will be executed on $v\in N[u]$. Let us denote by $\gamma''$ the resulting configuration after the transition where the said \textbf{Candidacy?} move has been executed on $v$. Since $u\in V_2$ we have $v\in V_1$ and using Lemma~\ref{bz-lem-main} $v \in I_{\gamma'''}$ with probability at least $\frac{1}{(\Delta+1)e}$ and $v \not\in I_{\gamma}$ by the same argument as above. Thus, since $I_{\gamma'''} \subseteq I_{\gamma'}$, the probability that $I_{\gamma} \subsetneq I_{\gamma'}$ is at least $\frac{1}{(\Delta+1)e}$.
\end{itemize}
~\\
In every case, the property is true, thus the lemma holds. 
\end{proof}

We will use the notation $\alpha = \frac{1}{(\Delta+1)e}$ to simplify the formulas in the remainder of the paper.

\begin{lemma} \label{bz-lem-speed}
For any $p\in [0,1[$. From any degree-stabilized configuration $\gamma$, the algorithm is self-stabilizing for a configuration $\gamma'$ where $I_{\gamma'}$ is a maximal independent set of $V_2 \cup I_{\gamma'}$. The time for this to happen is less than  $\max\left( -\alpha^{-2}\ln p,\frac{\sqrt{2}}{\sqrt{2}-1}\frac{n}{\alpha} \right)$ rounds with probability at least $1-p$.
\end{lemma}

\begin{proof}
Consider a degree-stabilized configuration $\gamma_0$, and $\Omega$ the set of all complete executions of the algorithm starting in configuration $\gamma_0$. The probability measure $\mathbb{P}$ is the one induced by the daemon choosing the worst possibility for us at every step.

Consider for $i \in \Nat$ the random variables $X_i$ that denote the configuration after the $i$-th round has ended ($X_0$ is the constant random variable of value $\gamma_0$).

Consider $(\mathcal{F}_i)_{i\in \Nat}$ the natural filtration associated with $X_i$.

Consider the function $f : \gamma \mapsto |I_{\gamma}|$.

$Y_i = \mathds{1}_{f(X_i) - f(X_{i-1}) > 0}$ is the random variable with value $1$ if the size of $I$ increased in the $i$-th round, else $0$.

Consider the stopping time $\tau$ (random variable describing the number of rounds the algorithm takes to stabilize on $V_1$) defined by: $$\tau(\omega) = \inf\quickset{n\in \Nat | I \text{ does not change after round } n \text{ in execution } \omega}$$

As $Y_i$ has values in $\quickset{0;1}$, we have $\mathbb{P}(Y_i=1 | \mathcal{F}_{i-1}) = \mathbb{E}[Y_i | \mathcal{F}_{i-1}]$. Also, from Lemma~\ref{bz-lem-advance} we get $\mathbb{P}(Y_i=1 | \mathcal{F}_{i-1}) \geq \alpha \cdot \mathds{1}_{\tau \geq i-1}$. Thus, combining the two relations, we get:

\begin{equation} \label{eqincr}
\mathbb{E}[Y_i | \mathcal{F}_{i-1}] \geq \alpha \cdot \mathds{1}_{\tau \geq i-1}
\end{equation}

Consider $S_i = \displaystyle \sum_{k=1}^{i}Y_k$ the random variable representing the number of rounds with an increment. Since there cannot be more increments than there are nodes, we get:

\begin{equation} \label{eqS}
S_i \leq n
\end{equation}

Consider $A_i = \displaystyle \sum_{k=1}^{i} \mathbb{E}[Y_{k} | \mathcal{F}_{k-1}]$, the random variable representing the sum of the expected values of the increments at each step.

When $\tau > i$, every for very value of $k \in \segment{1}{i}$ we have $\mathds{1}_{\tau \geq k-1} = 1$. Then using (\ref{eqincr}) we get:

\begin{equation} \label{eqA}
\tau > i \Rightarrow A_i \geq i\alpha
\end{equation}

Consider then the random variable $M_i = \sum_{k=1}^{i} Y_k - \mathbb{E}[Y_k|\mathcal{F}_{k-1}]$ (do note that it is the same as the difference $S_i - A_i$).
\begin{align*}
\mathbb{E}\left[ M_{i+1}|\mathcal{F}_i \right] &= \mathbb{E} \left[ \sum_{k=1}^{i+1} Y_k - \mathbb{E}[Y_k|\mathcal{F}_{k-1}] \middle| \mathcal{F}_i \right] \\
&= \mathbb{E} \left[ M_i + Y_{i+1} - \mathbb{E}[Y_{i+1}|\mathcal{F}_{i}] \middle| \mathcal{F}_i \right] \\
&= \mathbb{E} \left[ M_i \middle| \mathcal{F}_i \right]+ \mathbb{E} \left[Y_{i+1}\middle| \mathcal{F}_i \right] - \mathbb{E} \left[\mathbb{E}[Y_{i+1}|\mathcal{F}_{i}] \middle| \mathcal{F}_i \right] \\
&= M_i + \mathbb{E} \left[Y_{i+1} \middle| \mathcal{F}_i \right] - \mathbb{E} \left[Y_{i+1} \middle| \mathcal{F}_i \right] \\
&= M_i
\end{align*}

Thus $(M_i)_{i\in \Nat}$ is a martingale with respect to the filtration $(\mathcal{F}_i)_{i\in \Nat}$.

We also have $|M_{i+1} - M_i| = |Y_{i+1} - \mathbb{E}[Y_{i+1} | \mathcal{F}_i]| \leq \max(Y_{i+1},\mathbb{E}[Y_{i+1} | \mathcal{F}_i]) \leq 1$.

Thus, by Azuma inequality:
\begin{equation}\label{eqM}
\forall \beta \leq 0, \mathbb{P}(M_i \leq \beta) \leq e^{-\frac{2\beta^2}{i}}
\end{equation}

Then using (\ref{eqS}) and (\ref{eqA}) we get $\tau >i \Rightarrow S_i - A_i \leq n - i\alpha $,  \textit{i.e.} $\tau >i \Rightarrow M_i \leq n - i\alpha $

Thus we have $\mathbb{P}(\tau >i) \leq \mathbb{P}(M_i \leq n - i\alpha)$ and for $i \geq \frac{n}{\alpha}$ we can apply (\ref{eqM}) to get :

$\mathbb{P}(\tau >i) \leq e^{-\frac{2(n-i\alpha)^2}{i}}$

For $i\geq \frac{\sqrt{2}}{\sqrt{2}-1}\frac{n}{\alpha}$ (it implies that $i \geq \frac{n}{\alpha}$) we have $\frac{1}{2}(i\alpha)^2 \leq (n-i\alpha)^2$, which give for such $i$ :

$\mathbb{P}(\tau > i) \leq e^{-i\alpha^2}$

For $i \geq -\alpha^{-2}\ln p $, we have $e^{-i\alpha^2} \leq p$

Mixing the two above inequalities, when $i\geq \max\left( -\alpha^{-2}\ln p,\frac{\sqrt{2}}{\sqrt{2}-1}\frac{n}{\alpha} \right)$, we get:

$\mathbb{P}(\tau > i) \leq p$

Which concludes the proof.
\end{proof}

\begin{theorem}
For any $p\in [0,1[$. From any  configuration $\gamma$, the algorithm is self-stabilizing for a configuration $\gamma'$ where $I_{\gamma'}$ is a maximal independent set of $V_{2} \cup  I_{\gamma'}$, and reach such a configuration in $1+\max\left( -\alpha^{-2}\ln p,\frac{\sqrt{2}}{\sqrt{2}-1}\frac{n}{\alpha} \right)$ rounds or less with probability $1-p$.
\end{theorem}

\begin{proof}
From Lemma~\ref{bz-lem-stab-1round} we reach a degree-stabilized configuration after at most one round.
Then from that configuration, we apply Lemma~\ref{bz-lem-speed}. 
\end{proof}

\section{In an Anonymous System under the Adversary Daemon}
\label{sec:anonymous}

\subsection{The algorithm}

The algorithm builds a maximal independent set represented by a local variable~$s$.

We could have used the previous algorithm, as fairness was only needed to contain Byzantine influence. But the complexity would have been something proportional to $\Delta n^2$, and as we will prove, we can do better than that.

We keep the idea of having nodes making candidacy and then withdrawing if the candidate's situation is not correct. We still do not have identifiers, and so we do need a probabilistic tie-break. But contrary to the Byzantine case, we move the probabilities to the \textbf{Withdrawal} rule: a non-candidate node with no candidate neighbor will always become a candidate when activated, but a candidate node with a candidate neighbor will only withdraw with probability $\frac{1}{2}$ when activated.
\begin{algo}
Any node $u$ has a single local variable $s_u \in \quickset{\bot,\top}$ and may make a move according to one of the following rules: \\
\textbf{(Candidacy)} $(s_u=\bot)\wedge(\forall v \in N(u), s_v = \bot) \rightarrow s_u := \top$\\
\textbf{(Withdrawal?)} $(s_u=\top)\wedge(\exists v \in N(u), s_v =\top) \rightarrow $ if $Rand(\frac{1}{2})=1$ then $s_u := \bot$\\
\end{algo}

The following lemma guarantees the correctness of the algorithm.

\begin{lemma} \label{lem:adv-stab}
A configuration $\gamma$ is stable if and only if $\beta(\gamma)$ is a maximal independent set of $G$.
\end{lemma}

\begin{proof}
Suppose $\gamma$ is not stable.
\begin{itemize}
\item If \textbf{Withdrawal?} is activable on a node $u$, it means that $u \not\in \beta(\gamma)$ as $u$ has a neighbor with $s$-value $\top$. But as $u$ has $s$-value $\top$ it also means that none of its neighbors is in $\beta(\gamma)$, thus $\beta(\gamma) \cup \quickset{u}$ is an independent set greater than $\beta(\gamma)$.
\item Else, \textbf{Candidacy} is activable on a node $u$. It means that $u \not\in \beta(\gamma)$ as $s_u=\bot$. It also means that none of its neighbors is in $\beta(\gamma)$ since they all have $s$-value $\bot$, thus $\beta(\gamma) \cup \quickset{u}$ is an independent set greater than $\beta(\gamma)$.
\end{itemize}
Thus, by contraposition, if $\beta(\gamma)$ is a maximal independent set of $G$, then $\gamma$ is stable.

Suppose now that $\gamma$ is stable. As \textbf{Withdrawal?} cannot be enabled on any node, any node with $s$-value $\top$ is in $\beta(\gamma)$. But then, as \textbf{Withdrawal?} cannot be enabled on any node, every node has a node with $s$-value $\top$, and thus a node of $\beta(\gamma)$ in its closed neighborhood. $\beta(\gamma)$ is then a maximal independant set of $G$. 
\end{proof}

Now that we have correctness, we want to prove the convergence of our algorithm.
The idea that we will develop is that we can give a non-zero lower bound on the probability that a connected component of candidate nodes eventually collapses into at least one definitive member of the independent set. As every transition with \textbf{Candidacy} moves makes such scenes appear, and \textbf{Withdrawal?} moves make those collapse into at least one member of the independent set with non-zero probability, this should converge towards a maximal independent set.

Looking at the guards of the rules of the algorithm, we observe the following fact:
\begin{fact}\label{fact:basic}
In a configuration $\gamma$, when \textbf{Withdrawal?} is enabled on a node, \textbf{Candidacy} is not enabled on any of its neighbors.
\end{fact}

Given a configuration $\gamma$, we define $\beta(\gamma) = \quickset{u \in V \mid s_u {=} \top \wedge \forall v \in N(u), s_v {=} \bot}$. Note that $\beta(\gamma)$ is always an independent set since two distinct members cannot be neighbors (as they have both $s$-value $\top$).

\begin{lemma} \label{adv-lem-betagrows}
If $\gamma \rightarrow \gamma'$, then $\beta(\gamma) \subseteq \beta(\gamma')$.
\end{lemma}

\begin{proof}
For every $u\in \beta(\gamma)$, we have by definition:
\begin{enumerate}
\item $s_u^{\gamma}=\top$; 
\item for every neighbor $u$ of $v$, $s_u^{\gamma}=\bot$.
\end{enumerate}
Point~1 implies that \textbf{Candidacy} is not enabled on any vertex of $N[u]$, and Point~2 implies that \textbf{Withdrawal?} is not enabled on any vertex of $N[u]$.

For every vertex $v\in N[u]$, since no rule is enabled on $v$, we have $s_v^{\gamma} = s_v^{\gamma'}$. Thus, $u \in \beta(\gamma')$.
\end{proof}

Now that we know that $\beta$ cannot shrink, it remains to prove that it does grow, and within a reasonable amount of time.
We are introducing the following concepts to this end:

\begin{definition}
A set of node $\cX \subseteq V$ is said to be a \motnouveau{candidate set} of a configuration $\gamma$ if $\forall u \in \cX, (s_u^{\gamma} = \top) \wedge \forall v \in N(u), s_v^{\gamma} = \top \Rightarrow v\in \cX$. We will say that it is a \motnouveau{connected} candidate set when $\cX$ is a connected component of $G$.
\end{definition}

\subsection{The proof}
First, we transform all executions into executions having suitable properties, using the following two lemmas for the simplicity of the proof.

\begin{lemma} \label{lem:transformation1}
Let $\gamma_1$ be a configuration and $t$ a valid set of moves in $\gamma_1$. Let also $t_w \subset t$ (resp. $t_c$) the set of \textbf{Withdrawal?} moves (resp \textbf{Candidacy} moves) in $t$.

Then making two successive transitions using the set of moves $t_w$ and then $t_c$ is equivalent to making a transition with the set of moves $t$ (in the sense that the distribution of probability on the values of node variables is the same).

\textit{i.e.} we may replace $\gamma_{1} \xrightarrow{t} \gamma_2$ with $\gamma_{1} \xrightarrow{t_{c}} \gamma' \xrightarrow{t_{w}} \gamma_2$.
\end{lemma}

\begin{proof}
Consider $\gamma_2$ the random variable describing the state of the configuration attained from $\gamma_1$ by executing the set of moves $t$.

    Let $W$ (resp. $C$) be the set of all the nodes appearing in a \textbf{Withdrawal?} (resp. \textbf{Candidacy}) move of $t$. 

There is only one possible $\gamma'$ such that $\gamma_1 \xrightarrow{t_c} \gamma'$ since the rule \textbf{Candidacy} is deterministic, consider this configuration $\gamma'$.

Due to Fact \ref{fact:basic}, no node in $W$ neighbors a node of $C$. Thus,  nodes in $W$ and their neighbors did not change their local variables in the transition $\gamma_{1} \xrightarrow{t_{c}} \gamma'$. Then, \textbf{Withdrawal?} is enabled on every node of $W$ in $\gamma'$ and $t_w$ is a valid set of moves in $\gamma'$.

Consider then the random variable $\gamma_2'$ such that $\gamma' \xrightarrow{t_w} \gamma_2'$.

\begin{itemize}
\item For $u \in C$, the transitions are deterministic and we have $s_u^{\gamma_2'} = \top = s_u^{\gamma_2}$
\item For $u \in W$, note that as long as \textbf{Withdrawal?} is executed, its results always follow the same probability law. Thus, as nodes that executed \textbf{Withdrawal} are the same in $\gamma' \xrightarrow{t_w} \gamma_2'$ and $\gamma_1 \xrightarrow{t} \gamma_2$, the probability distribution of $s_u$ is the same in $\gamma_2$ and $\gamma_2'$.
\item For nodes that are in neither $C$ nor $W$, they didn't execute any rules in the considered transitions, and thus they have the same $s$-value in every configuration we considered.
\end{itemize}

Thus, $\gamma_2$ has the same probability distribution as $\gamma_2'$, and hence the result.
\end{proof}

\begin{lemma} \label{lem:transformation2}
Let $\gamma_1$ be a configuration, $t$ a valid set of moves in $\gamma_1$ containing only \textbf{Withdrawal?} moves, and $A$ a candidate set of $\gamma_1$. Let also $t_A \subset t$ be the set of moves on nodes of $A$ in $t$.

Then making two successive transitions using the set of moves $t_A$ and then $t \setminus t_A$ is equivalent to making a transition with the set of moves $t$ (in the sense that the distribution of probability on the values of node variables is the same).

\textit{i.e.} we may replace $\gamma_{1} \xrightarrow{t} \gamma_2$ with $\gamma_{1} \xrightarrow{t_A} \gamma' \xrightarrow{t \setminus t_A} \gamma_2$.
\end{lemma}

\begin{proof}
The proof follow the same pattern as the one of Lemma~\ref{lem:transformation2} and rely in the same way on the fact that a rule -here \textbf{Withdrawal}- applied on a node $u$ have the same distribution of probability of $s$-value after the transition whatever the starting configuration of the transition.

The only thing left to prove is that $t \setminus t_A$ is a valid set of moves in any $\gamma'$ such that $\gamma_{1} \xrightarrow{t_A} \gamma'$.

Let's denote by $B$ the set of the nodes that appear in moves of $t \setminus t_A$. 

Consider $u \in B$. As \textbf{Withdrawal?} is enabled on $u$ in $\gamma_1$ we know that $s_u^{\gamma_1} = \top$. Then, by definition of a candidate set, $u$ cannot be a neighbor of a node in $A$, and as a result no node of $N[u]$ is activated in the transition $\gamma_{1} \xrightarrow{t_A} \gamma'$. The closed neighborhood of $u$ having the same $s$-values in $\gamma_1$ and $\gamma'$, \textbf{Withdrawal?} must be enabled on $u$ in $\gamma'$.

It is then the case for every node of $B$, and thus $t \setminus t_A$ is a valid set of moves in $\gamma'$.
\end{proof}

Using those lemmas, we will now only consider executions satisfying the following three properties:
\begin{enumerate}
\item Each transition is composed of only one move type.
\item For each transition $ \gamma_{i-1} \xrightarrow{t_i} \gamma_{i}$ containing $\textbf{Withdrawal?}$ moves,  
  only nodes of the same connected candidate set in  $ \gamma_{i-1}$ execute a move.
\end{enumerate}
 
Using Lemmas~\ref{lem:transformation1}~and~\ref{lem:transformation2}, all executions can be transformed into executions satisfying the two last properties without increasing the number of moves.

From now on, all the executions we consider satisfy the two previous properties.
 
Let $\mathcal E$ be an execution  $\gamma_0 \xrightarrow{t_1} \gamma_1 \dotsm \gamma_{i-1} \xrightarrow{t_i} \gamma_{i}\dotsm$.
 
We say that a candidate set $\cX$ of $\gamma$ is \motnouveau{alive} in $\gamma$ if \textbf{Withdrawal} is enabled on at least one node of $A$. Note that, for a connected candidate set, it is equivalent to having cardinality at least $2$.

In a transition $\gamma \xrightarrow{t} \gamma'$, we say that an alive candidate set $\cX$ of $\gamma$ \motnouveau{vanishes} if $\forall u\in \cX$, $s_u^{\gamma'} = \bot$ or $u\in \beta_{\gamma'}$.

\begin{lemma} \label{lem:xivanishes}
Suppose that $\cX$ is an alive connected candidate set of $\gamma$, and that $t$ is a valid set of moves in $\gamma$ that involve nodes of $\cX$ (recall that we transformed our execution to make withdrawal transitions only act on a given connected candidate set).

Then, if $\gamma'$ is the random variable describing the state of the system after performing the set of move $t$ from $\gamma$, the probability that $\beta(\gamma') \setminus \beta(\gamma) \not= \emptyset$ knowing that $\cX$ vanishes is at least $\frac{2}{3}$.
\end{lemma}

\begin{proof}If there is one node $u$ in $\cX$ that is not activated in the transition, we know that $s_u^{\gamma'} = s_u^{\gamma} = \top$. As $\cX$ is supposed to vanish in the transition, it implies that $u \in \beta(\gamma')$.

Else, every node of $\cX$ has performed \textbf{Withdrawal} in the transition.

Consider the distribution of the $s$-values after performing the set of moves $t$ in configuration $\gamma$ (without the condition on the candidate set vanishing). Let's write $k=|\cX|$. As every node of $\cX$ flips a coin independently, every outcome has the same probability $\left(\frac{1}{2}\right)^{k}$. In particular,
\begin{itemize}
\item The event where every node of $\cX$ change their $s$-value to $\bot$ in the transition.
\item Each of the $k$ events where every node of $\cX$ change their $s$-value to $\bot$ in the transition except exactly one.
\end{itemize}

Those events are compatible with the constraint of the candidate set vanishing. Every other event is either incompatible with the constraint, or a situation where at least two nodes of $\cX$ are in $\beta(\gamma')$, as at least two nodes have $s$-value $\top$ in $\gamma'$ and $\cX$ is supposed to vanish in the transition.

Thus, if denote by $\lambda$ the sum of the probability of the events that are compatible with the constraint of the candidate set vanishing, we can say that $\lambda \geq (k+1)\left(\frac{1}{2}\right)^k$, and thus the probability to have $\beta(\gamma_{j^{*}}) \setminus \beta(\gamma_{j^{*}-1}) \not= \emptyset$ with the condition is $1-\frac{\left(\frac{1}{2}\right)^k}{\lambda} \geq 1-\frac{\left(\frac{1}{2}\right)^k}{(k+1)\left(\frac{1}{2}\right)^k}=1-\frac{1}{k+1} = \frac{k}{k+1}$.

Moreover, $\cX$ is not empty by hypothesis. Consider $u \in \cX$, it has at least a neighbor $v$ such that $s_v^{\gamma} = \top$ as otherwise we would have $u \in \beta(\gamma)$. But then we have $v \in \cX$, as $\cX$ is a candidate set. Thus $k \geq 2$.

Thus the wanted probability is at least $\frac{2}{3}$.
\end{proof}

Now that we know that every time that an alive connected candidate set vanishes, $\beta$ progresses with good probability, it remains to guarantee that it ever happens. To do so, we will need the following lemmas about the life and death of candidate sets.

\begin{lemma} \label{lem:newalive}
If $\gamma \xrightarrow{t} \gamma'$ is a transition with $t$ containing only \textbf{Candidacy} moves, either $\beta(\gamma) \subsetneq \beta(\gamma')$ or $\exists \cX$ an alive connected candidate set of $\gamma'$ such that $\forall u \in \cX, s_u^{\gamma} = \bot$.
\end{lemma}

\begin{proof}
As a transition must contain at least one move, there is at least one node $u$ appearing in $t$.

If no neighbor of $u$ appears in $t$, we know that no neighbor of $u$ changes its $s$-value in the transition. But since \textbf{Candidacy} is enabled on $u$ in $\gamma$ we know that their $s$-value in $\gamma$ is $\bot$. Thus, $s_u^{\gamma'}=\top$ (since $u$ executed \textbf{Candidacy} in the transition), and every one of its neighbors has $s$-value $\bot$ in $\gamma'$, which is the definition for $u \in \beta(\gamma')$. And since \textbf{Candidacy} is enabled on $u$ in $\gamma$, $s_u^{\gamma}=\bot$, thus $u \not\in \beta(\gamma')$. Thus $\beta(\gamma) \subsetneq \beta(\gamma')$ (recall that from Lemma~\ref{adv-lem-betagrows} $\beta$ cannot lose nodes in a transition).

Suppose now that at least a neighbor of $u$ appears in $t$. Then consider the set $\cX$ of the nodes that are accessible from $u$ in the subgraph induced by the nodes appearing in $t$. Every node $v \in \cX$ executes \textbf{Candidacy} in the transition, thus $s_v^{\gamma}=\bot$ and $s_v^{\gamma'}=\top$. Moreover, for every neighbor of a node of $\cX$ not in $\cX$, the condition of \textbf{Candidacy} implies that it has $s$-value $\bot$ in $\gamma$, and thus in $\gamma'$ too since it did not perform a rule in the transition. Thus $\cX$ is a candidate set of $\gamma'$, connected by construction, and of cardinality at least $2$ by hypothesis, thus alive.
\end{proof}

\begin{lemma} \label{lem:stayingalive}
If $\gamma \xrightarrow{t} \gamma'$ is a transition and $\cX$ is an alive connected candidate set of $\gamma$, one of those is true:
\begin{itemize}
\item $\cX$ vanishes in the transition.
\item $\exists \cX'$ an alive connected candidate set of $\gamma'$ such that $\cX' \subseteq \cX$
\end{itemize}
\end{lemma}

\begin{proof}
If $t$ is a set of \emph{Candidacy} moves, or a set of \textbf{Withdrawal?} moves on a connected candidate set different from $\cX$, then the transition does not affect neither nodes of $\cX$ or their neighbor thus $\cX \subseteq \cX$ is still an alive connected candidate set of $\gamma'$.

Else, $t$ is a set of \textbf{Withdrawal?} moves on nodes of $\cX$. Suppose $\cX$ does not vanish in the transition. By definition of vanishing, it means that there exists a node $u \in \cX$ such that $s_u^{\gamma'} = \top$ and $u\not\in \beta(\gamma')$. It implies that there exists $v \in N(u)$ such that $s_v^{\gamma'} = \top$. But then, as $t$ contains only \textbf{Withdrawal?} moves, $s_v^{\gamma}=\top$. By definition of a candidate set, we must have $v \in \cX$. Now consider $\cX'$ the set of nodes in the connected component of $u$ in the subgraph of $G$ induced by the nodes with $s$-value $\top$ in $\gamma'$. As only \textbf{Withdrawal?} moves have been executed in the transition, this imply that $\cX'$ is a subset of the set of nodes in the connected of the connected component of $u$ in the subgraph of $G$ induced by the nodes with $s$-value $\top$ in $\gamma$, which is exactly $\cX$. Moreover, if $w\not\in \cX'$ is a neighbor of a node in $\cX'$, by definition of $\cX'$ it is such that $s_w^{\gamma'}= \bot$. Hence, $cX'$ is a connected candidate set of $\gamma'$, alive since its cardinality is at least $2$.\end{proof}

\begin{lemma} \label{lem:maxalive}
There cannot be more than $\frac{n}{2} - |\beta(\gamma)|$ different connected candidate sets alive in a configuration $\gamma$.
\end{lemma}

\begin{proof}
An alive candidate set must have at least two nodes. Two candidate sets that share a node must be the same candidate set, as the definition of a candidate set prevents having an $s$-value $\top$ in the neighborhood of a candidate set.
\end{proof}

\begin{lemma}\label{lem:execends}
The execution ends with probability $1$.
\end{lemma}

\begin{proof}
Every transition is either a transition with only \textbf{Withdrawal} moves or only \textbf{Candidacy} moves.

Note that, from Lemma~\ref{lem:stayingalive}, the number of alive connected components cannot shrink without one vanishing, and shrinks by exactly one when it does happen.

Every transition with \textbf{Withdrawal} moves that do not make an alive connected candidate set vanish makes an alive connected candidate set shrink with probability at least $\frac{1}{2}$. Moreover when it makes one of those disappear is increase the size of $\beta$ with probability at least $\frac{2}{3}$.

Every transition with \textbf{Candidacy} moves either makes a new alive candidate set appear, or increases the size of $\beta$. But as there can not be more than $\frac{n}{2} - |\beta|$ alive candidate sets, there cannot be a transition that makes a new alive candidate set appear while the number of alive candidate sets is maximal, and thus one must vanish before it becomes possible again (which will increase the size of $\beta$ with probability at least $\frac{2}{3}$).

By induction on the number of alive connected candidate set and their size, as long as \textbf{Candidacy} or \textbf{Withdrawal} are enabled, $\beta$ will grow within a finite time. As its size cannot be greater than the number of nodes, the execution must end.
\end{proof}

Now, we may suppose that every alive candidate set vanishes in finite time, and use our Lemma~\ref{lem:xivanishes}.

\begin{lemma} \label{lem:numberofvanish}
Let $1> p \geq 0$.  Any execution has at most  $\max\left( -\frac{9}{4}\ln p,\frac{\sqrt{2}}{\sqrt{2}-1}\frac{3n}{2} \right)$ transitions in which a connected candidate set vanishes with probability at least $1-p$.
\end{lemma}

\begin{proof}
It is the same as the proof of Lemma~\ref{bz-lem-speed}, but the event we track is the vanishing of a connected candidate set (which happens in finite time with probability $1$ from Lemma~\ref{lem:execends}), the lower bound on the probability of the interesting event to happen when such a vanishing happens is $\frac{2}{3}$ from Lemma~\ref{lem:xivanishes}.
\end{proof}

\begin{theorem} 
Consider $p \in \left]0,1\right]$. From any configuration, the algorithm is self-stabilizing for a configuration $\gamma$ in which 
$\beta(\gamma)$ is a maximal independent set,  
with at most $\mathcal O(n^2) $ moves
with probability at least $1-p$.
\end{theorem}

\begin{proof}
Consider $p' \in \left]0,1\right]$, $p' = \frac{p}{2}$, and $\lambda = \max\left( -\frac{9}{4}\ln p',\frac{\sqrt{2}}{\sqrt{2}-1}\frac{3n}{2} \right)$.
We denote by $x$ the number of vanishing connected candidate sets in the execution. Using Lemma~\ref{lem:numberofvanish}, we know that $x> \lambda$ with probability at most $p'$.

First, we count the number of moves of type \textbf{(u,Candidacy)}. For each transition in which a connected candidate set vanishes, those nodes have had their $s$-value set to $\top$ by the last transition containing \textbf{Candidacy} moves on the said candidate set (or it does not exist and they had $s$-value $\top$ in the initial state). This makes for at most $n$ \textbf{Candidacy} moves for every connected candidate set that vanish, thus $xn$ in total.

To count for \textbf{Withdrawal?} moves, we will first count the number of successful moves (meaning the $s$-value of the node it is applied to changes in the transition). For every connected candidate set that vanishes, as it may only shrink from the time it appears due to Lemma~\ref{lem:stayingalive}, it takes at most $n$ successful \textbf{Withdrawal?} moves to make it disappear. Note that any \textbf{Withdrawal?} move acts on an alive connected candidate set, thus every \textbf{Withdrawal?} participates in the shrinking of an alive connected candidate set. Then, in total, we need at most $xn$ successful \textbf{Withdrawal?} moves to have $x$ vanishings of connected candidate sets.

Each \textbf{Withdrawal?} move has probability $\frac{1}{2}$ to be successful, a probability which does not depend on anything but the fact that rule \textbf{Withdrawal?} is applied, and are thus the outcome of different \textbf{Withdrawal?} moves are independent. The number of successful \textbf{Withdrawal?} is then connected to the total number of \textbf{Withdrawal?} by a binomial law with parameters $\lambda$ and $\frac{1}{2}$, where $\lambda$ is the number of \textbf{Withdrawal?} moves. Using the standard bound obtained by Hoeffding's inequality for the binomial law, the probability of having $\lambda n-1$ or fewer successful \textbf{Withdrawal?} after $2(\lambda n + \sqrt{- \lambda n\log p' } - 1)$ \textbf{Withdrawal?} moves is at most $p$. This means that the probability of not having finished the execution after $2(xn + \sqrt{-xn\log p' } - 1)$ \textbf{Withdrawal?} moves is at most $p'$.

Then, using the union bound, the probability that $x$ is such that $x > \lambda$ or the number of \textbf{Withdrawal?} moves in the execution is greater than $2(\lambda n + \sqrt{- \lambda n\log p' } - 1)$ is at most $2p'=p$.

Thus, the probability of converging in less than $2(\lambda n + \sqrt{- \lambda n\log p' } - 1) + \lambda n = O(n^2)$ moves in total is at least $1-p$, which concludes the proof.
\end{proof}

\bibliography{MIS}{}

@article{SHI200477,
title = {An anonymous self-stabilizing algorithm for 1-maximal independent set in trees},
journal = {Information Processing Letters},
volume = {91},
number = {2},
pages = {77-83},
year = {2004},
issn = {0020-0190},
doi = {https://doi.org/10.1016/j.ipl.2004.03.010},
url = {https://www.sciencedirect.com/science/article/pii/S0020019004000985},
author = {Zhengnan Shi and Wayne Goddard and Stephen T. Hedetniemi}
}

@article{Alon1986,
title = {A fast and simple randomized parallel algorithm for the maximal independent set problem},
journal = {Journal of Algorithms},
volume = {7},
number = {4},
pages = {567-583},
year = {1986},
issn = {0196-6774},
doi = {https://doi.org/10.1016/0196-6774(86)90019-2},
url = {https://www.sciencedirect.com/science/article/pii/0196677486900192},
author = {Noga Alon and László Babai and Alon Itai}
}

@INPROCEEDINGS{balliu2019lower,
  author={Balliu, Alkida and Brandt, Sebastian and Hirvonen, Juho and Olivetti, Dennis and Rabie, Mikaël and Suomela, Jukka},
  booktitle={2019 IEEE 60th Annual Symposium on Foundations of Computer Science (FOCS)}, 
  title={Lower Bounds for Maximal Matchings and Maximal Independent Sets}, 
  year={2019},
  volume={},
  number={},
  pages={481-497},
  doi={10.1109/FOCS.2019.00037}}

@inproceedings{barenboim2018locally,
author = {Barenboim, Leonid and Elkin, Michael and Goldenberg, Uri},
year = {2018},
month = {07},
pages = {437-446},
title = {Locally-Iterative Distributed ($\Delta$ + 1): -Coloring below Szegedy-Vishwanathan Barrier, and Applications to Self-Stabilization and to Restricted-Bandwidth Models},
isbn = {978-1-4503-5795-1},
journal = {PODC '18: Proceedings of the 2018 ACM Symposium on Principles of Distributed Computing},
booktitle={Proceedings of the 2018 ACM Symposium on Principles of Distributed Computing},
doi = {10.1145/3212734.3212769}
}

@article{benreguia2021selfstabilizing,
  author    = {Badreddine Benreguia and
               Hamouma Moumen and
               Soheila Bouam and
               Chafik Arar},
  title     = {Self-stabilizing Algorithm for Maximal Distance-2 Independent Set},
  journal   = {CoRR},
  volume    = {abs/2101.11126},
  year      = {2021},
  url       = {https://arxiv.org/abs/2101.11126},
  eprinttype = {arXiv},
  eprint    = {2101.11126},
  bibsource = {dblp computer science bibliography, https://dblp.org}
}

@article{censor2020derandomizing,
author = {Censor-Hillel, Keren and Parter, Merav and Schwartzman, Gregory},
title = {Derandomizing Local Distributed Algorithms under Bandwidth Restrictions},
year = {2020},
issue_date = {Jun 2020},
publisher = {Springer-Verlag},
address = {Berlin, Heidelberg},
volume = {33},
number = {3–4},
issn = {0178-2770},
url = {https://doi.org/10.1007/s00446-020-00376-1},
doi = {10.1007/s00446-020-00376-1},
journal = {Distrib. Comput.},
month = {jun},
pages = {349–366},
numpages = {18}
}

@inproceedings{codevi2006,
author = {Cournier, Alain and Devismes, Stephane and Villain, Vincent},
title = {Snap-Stabilizing PIF and Useless Computations},
year = {2006},
isbn = {0769526128},
publisher = {IEEE Computer Society},
address = {USA},
url = {https://doi.org/10.1109/ICPADS.2006.100},
doi = {10.1109/ICPADS.2006.100},
booktitle = {Proceedings of the 12th International Conference on Parallel and Distributed Systems - Volume 1},
pages = {39–48},
numpages = {10},
series = {ICPADS '06}
}

@article{dijkstra74,
author = {Dijkstra, Edsger W.},
title = {Self-Stabilizing Systems in Spite of Distributed Control},
year = {1974},
issue_date = {Nov. 1974},
publisher = {Association for Computing Machinery},
address = {New York, NY, USA},
volume = {17},
number = {11},
issn = {0001-0782},
url = {https://doi.org/10.1145/361179.361202},
doi = {10.1145/361179.361202},
journal = {Commun. ACM},
month = {nov},
pages = {643–644},
numpages = {2}
}

@InProceedings{doismo97,
author="Dolev, Shlomi
and Israeli, Amos
and Moran, Shlomo",
editor="Toueg, Sam
and Spirakis, Paul G.
and Kirousis, Lefteris",
title="Uniform dynamic self-stabilizing leader election",
booktitle="Distributed Algorithms",
year="1992",
publisher="Springer Berlin Heidelberg",
address="Berlin, Heidelberg",
pages="167--180",
isbn="978-3-540-46789-2"
}

@article{gao2017novel,
author = {Gao, Xiaofeng and Zhu, Xudong and Li, Jun and Wu, Fan and Chen, Guihai and Du, Ding-Zhu and Tang, Shaojie},
title = {A Novel Approximation for Multi-Hop Connected Clustering Problem in Wireless Networks},
year = {2017},
issue_date = {August 2017},
publisher = {IEEE Press},
volume = {25},
number = {4},
issn = {1063-6692},
url = {https://doi.org/10.1109/TNET.2017.2690359},
doi = {10.1109/TNET.2017.2690359},
journal = {IEEE/ACM Trans. Netw.},
month = {aug},
pages = {2223–2234},
numpages = {12}
}

@inproceedings{ghaffari2016improved,
author = {Ghaffari, Mohsen},
title = {An Improved Distributed Algorithm for Maximal Independent Set},
year = {2016},
isbn = {9781611974331},
publisher = {Society for Industrial and Applied Mathematics},
address = {USA},
booktitle = {Proceedings of the Twenty-Seventh Annual ACM-SIAM Symposium on Discrete Algorithms},
pages = {270–277},
numpages = {8},
location = {Arlington, Virginia},
series = {SODA '16}
}

@inproceedings{ghaffari2021improved,
author = {Ghaffari, Mohsen and Grunau, Christoph and Rozho\v{n}, V\'{a}clav},
title = {Improved Deterministic Network Decomposition},
year = {2021},
isbn = {9781611976465},
publisher = {Society for Industrial and Applied Mathematics},
address = {USA},
booktitle = {Proceedings of the Thirty-Second Annual ACM-SIAM Symposium on Discrete Algorithms},
pages = {2904–2923},
numpages = {20},
location = {Virtual Event, Virginia},
series = {SODA '21}
}

@INPROCEEDINGS{goddard2003self,
  author={Goddard, W. and Hedetniemi, S.T. and Jacobs, D.P. and Srimani, P.K.},
  booktitle={Proceedings International Parallel and Distributed Processing Symposium}, 
  title={Self-stabilizing protocols for maximal matching and maximal independent sets for ad hoc networks}, 
  year={2003},
  volume={},
  number={},
  pages={14 pp.-},
  doi={10.1109/IPDPS.2003.1213302}}

@INPROCEEDINGS{grati2007,
  author={Gradinariu, Maria and Tixeuil, Sebastien},
  booktitle={27th International Conference on Distributed Computing Systems (ICDCS '07)}, 
  title={Conflict Managers for Self-stabilization without Fairness Assumption}, 
  year={2007},
  volume={},
  number={},
  pages={46-46},
  doi={10.1109/ICDCS.2007.95}}

@article{guellati2010survey,
author = {Guellati, Nabil and Kheddouci, Hamamache},
title = {A Survey on Self-Stabilizing Algorithms for Independence, Domination, Coloring, and Matching in Graphs},
year = {2010},
issue_date = {April, 2010},
publisher = {Academic Press, Inc.},
address = {USA},
volume = {70},
number = {4},
issn = {0743-7315},
url = {https://doi.org/10.1016/j.jpdc.2009.11.006},
doi = {10.1016/j.jpdc.2009.11.006},
journal = {J. Parallel Distrib. Comput.},
month = {apr},
pages = {406–415},
numpages = {10}
}

@incollection{hedetniemi2020self,
  title={Self-stabilizing domination algorithms},
  author={Hedetniemi, Stephen T},
  booktitle={Structures of Domination in Graphs},
  pages={485--520},
  year={2020},
  publisher={Springer}
}

@inproceedings{ikeda2002space,
author = {Ikeda, Michiyo and Kamei, Sayaka and Kakugawa, Hirotsugu},
year = {2002},
month = {01},
  booktitle={the Third International Conference on Parallel and Distributed Computing, Applications and Technologies (PDCAT)},
  pages={70--74},
title = {A Space-Optimal Self-Stabilizing Algorithm for the Maximal Independent Set Problem}
}

@techreport{johnen:hal-03138979,
  TITLE = {{Efficient self-stabilizing construction of disjoint MDSs in distance-2 model}},
  AUTHOR = {Johnen, Colette and Haddad, Mohammed},
  URL = {https://hal.archives-ouvertes.fr/hal-03138979},
  TYPE = {Research Report},
  INSTITUTION = {{Inria Paris, Sorbonne Universit{\'e} ; LaBRI, CNRS UMR 5800 ; LIRIS UMR CNRS 5205}},
  YEAR = {2021},
  MONTH = Feb,
  HAL_ID = {hal-03138979},
  HAL_VERSION = {v1},
}

@article{lamport82,
author = {Lamport, Leslie and Shostak, Robert and Pease, Marshall},
title = {The Byzantine Generals Problem},
year = {1982},
issue_date = {July 1982},
publisher = {Association for Computing Machinery},
address = {New York, NY, USA},
volume = {4},
number = {3},
issn = {0164-0925},
url = {https://doi.org/10.1145/357172.357176},
doi = {10.1145/357172.357176},
journal = {ACM Trans. Program. Lang. Syst.},
month = {jul},
pages = {382–401},
numpages = {20}
}

@INPROCEEDINGS{linial1987,
  author={Linial, Nathan},
  booktitle={28th Annual Symposium on Foundations of Computer Science (sfcs 1987)}, 
  title={Distributive graph algorithms Global solutions from local data}, 
  year={1987},
  volume={},
  number={},
  pages={331-335},
  doi={10.1109/SFCS.1987.20}}

@article{liu2017cooperative,
author = {Liu, Tao and Wang, Xiaodong and Zheng, Le},
year = {2017},
month = {03},
pages = {1-1},
title = {A Cooperative SWIPT Scheme for Wirelessly Powered Sensor Networks},
volume = {PP},
journal = {IEEE Transactions on Communications},
doi = {10.1109/TCOMM.2017.2685580}
}

@inproceedings{luby1986,
author = {Luby, M},
title = {A Simple Parallel Algorithm for the Maximal Independent Set Problem},
year = {1985},
isbn = {0897911512},
publisher = {Association for Computing Machinery},
address = {New York, NY, USA},
url = {https://doi.org/10.1145/22145.22146},
doi = {10.1145/22145.22146},
booktitle = {Proceedings of the Seventeenth Annual ACM Symposium on Theory of Computing},
pages = {1–10},
numpages = {10},
location = {Providence, Rhode Island, USA},
series = {STOC '85}
}

@book{peleg2000distributed,
author = {Peleg, David},
title = {Distributed Computing: A Locality-Sensitive Approach},
year = {2000},
isbn = {0898714648},
publisher = {Society for Industrial and Applied Mathematics},
address = {USA}
}

@inbook{rozhovn2020polylogarithmic,
author = {Rozho\v{n}, V\'{a}clav and Ghaffari, Mohsen},
title = {Polylogarithmic-Time Deterministic Network Decomposition and Distributed Derandomization},
year = {2020},
isbn = {9781450369794},
publisher = {Association for Computing Machinery},
address = {New York, NY, USA},
url = {https://doi.org/10.1145/3357713.3384298},
booktitle = {Proceedings of the 52nd Annual ACM SIGACT Symposium on Theory of Computing},
pages = {350–363},
numpages = {14}
}

@article{shukla1995observations,
author = {Shukla, Sandeep and Rosenkrantz, Daniel and Ravi, S.},
year = {1995},
month = {01},
pages = {},
title = {Observations on self-stabilizing graph algorithms for anonymous networks},
journal = {Proceedings of the Second Workshop on Self-Stabilizing Systems}
}

@article{tanaka2021self,
  title={A Self-stabilizing 1-maximal Independent Set Algorithm},
  author={Hideyuki Tanaka and Yuichi Sudo and Hirotsugu Kakugawa and Toshimitsu Masuzawa and Ajoy K. Datta},
  journal={Journal of Information Processing},
  volume={29},
  number={ },
  pages={247-255},
  year={2021},
  doi={10.2197/ipsjjip.29.247}
}

@article{turau2007linear,
author = {Turau, Volker},
title = {Linear Self-Stabilizing Algorithms for the Independent and Dominating Set Problems Using an Unfair Distributed Scheduler},
year = {2007},
issue_date = {July, 2007},
publisher = {Elsevier North-Holland, Inc.},
address = {USA},
volume = {103},
number = {3},
issn = {0020-0190},
url = {https://doi.org/10.1016/j.ipl.2007.02.013},
doi = {10.1016/j.ipl.2007.02.013},
journal = {Inf. Process. Lett.},
month = {jul},
pages = {88–93},
numpages = {6}
}

@inproceedings{turau2019making,
author = {Turau, Volker},
year = {2019},
month = {07},
pages = {309-324},
booktitle={International Colloquium on Structural Information and Communication Complexity},
title = {Making Randomized Algorithms Self-stabilizing},
isbn = {978-3-030-24921-2},
organization={Springer},
doi = {10.1007/978-3-030-24922-9_21}
}

@inproceedings{turau2006randomized,
author = {Turau, Volker and Weyer, Christoph},
title = {Randomized Self-Stabilizing Algorithms for Wireless Sensor Networks},
year = {2006},
isbn = {3540376585},
publisher = {Springer-Verlag},
address = {Berlin, Heidelberg},
url = {https://doi.org/10.1007/11822035_8},
doi = {10.1007/11822035_8},
booktitle = {Proceedings of the First International Conference, and Proceedings of the Third International Conference on New Trends in Network Architectures and Services Conference on Self-Organising Systems},
pages = {74–89},
numpages = {16},
location = {Passau, Germany},
series = {IWSOS'06/EuroNGI'06}
}

@inproceedings{nous21,
  author       = {Johanne Cohen and
                  Laurence Pilard and
                  Jonas S{\'{e}}nizergues},
  editor       = {Colette Johnen and
                  Elad Michael Schiller and
                  Stefan Schmid},
  title        = {Self-stabilization and Byzantine Tolerance for Maximal Independent
                  Set},
  booktitle    = {Stabilization, Safety, and Security of Distributed Systems - 23rd
                  International Symposium, {SSS} 2021, Virtual Event, November 17-20,
                  2021, Proceedings},
  series       = {Lecture Notes in Computer Science},
  volume       = {13046},
  pages        = {479--483},
  publisher    = {Springer},
  year         = {2021},
  url          = {https://doi.org/10.1007/978-3-030-91081-5\_33},
  doi          = {10.1007/978-3-030-91081-5\_33},
  bibsource    = {dblp computer science bibliography, https://dblp.org}
}
\end{document}